\documentclass[pre,aps,showkeys,showpacs,superscriptaddress,twocolumn,%
amsmath,amssymb]{revtex4} \usepackage{graphicx,amsthm,dsfont}

\newtheorem*{thm}{Theorem}
\newtheorem*{cor}{Corollary}
\newtheorem*{lem}{Lemma}
\newtheorem*{alg}{Algorithm}

\newcommand{\gl}[1]{\eqref{#1}}
\newcommand{\la}{\left\langle}
\newcommand{\ra}{\right\rangle}
\newcommand{\al}{\alpha}
\newcommand{\be}{\beta}

\newcommand{\df}{\Delta f}
\newcommand{\dfhat}{\widehat{\df}}
\newcommand{\dfhatopt}{\dfhat_{1-\al}}
\newcommand{\deltf}{\delta \! f}

\newcommand{\pf}{p_{0}}
\newcommand{\pr}{p_{1}}
\newcommand{\pa}{p_{\al}}
\newcommand{\Ua}{U_{\al}}
\newcommand{\Uhat}{\widehat{U}}
\newcommand{\Mhat}{\widehat{M}}
\newcommand{\Xhat}{\widehat{X}}

\newcommand{\alhat}{\widehat{\al}}
\newcommand{\behat}{\widehat{\be}}

\newcommand{\var}{\operatorname{Var}}
\newcommand{\varf}[1]{\var_{0}\!\! \left(#1\right)}
\newcommand{\varr}[1]{\var_{1}\!\! \left(#1\right)}
\newcommand{\varal}[1]{\var_{\al}\!\! \left(#1\right)}

\newcommand{\mean}[2]{\la #1 \ra_{\! #2}}
\newcommand{\meanf}[1]{\mean{#1}{0}}
\newcommand{\meanr}[1]{\mean{#1}{1}}

\begin{document}
\title{A characteristic of Bennett's acceptance ratio method}
\author{A.M.\,Hahn} \altaffiliation{Current address: Technische
  Universit\"at Berlin, Institut f\"ur Theo\-re\-ti\-sche Physik,
  10623 Berlin, Germany} \affiliation{Institut f\"ur Physik, Carl von
  Ossietzky Universit\"at, 26111 Oldenburg, Germany} \author{H.\,Then}
\affiliation{Institut f\"ur Physik, Carl von Ossietzky Universit\"at,
  26111 Oldenburg, Germany}

\begin{abstract}
  A powerful and well-established tool for free-energy estimation is
  Bennett's acceptance ratio method. Central properties of this estimator,
  which employs samples of work values of a forward and its time
  reversed process, are known: for given sets of measured work values,
  it results in the best estimate of the free-energy difference in the large
  sample limit. Here we state and prove a further characteristic
  of the acceptance ratio method: the convexity of its mean square error.
  As a two-sided estimator, it depends on the ratio of the numbers of
  forward and reverse work values used. Convexity of its mean
  square error immediately implies that there exists an unique optimal ratio
  for which the error becomes minimal. Further, it yields
  insight into the relation of the acceptance ratio method
  and estimators based on the Jarzynski equation.
  As an application, we study the performance of a dynamic strategy
  of sampling forward and reverse work values.
\end{abstract}

\pacs{05.40.-a, 05.70.Ln} \keywords{fluctuation theorem,
  nonequilibrium thermodynamics}

\maketitle

\section{Introduction}\label{sec:1}

A quantity of central interest in thermodynamics and statistical physics
is the (Helmholtz) free-energy, as it determines the equilibrium properties
of the system under consideration. In practical applications, e.g.\ 
drug design, molecular association, thermodynamic stability, and binding
affinity, it is usually sufficient to know free-energy differences.
As recent progress in statistical physics has shown, free-energy
differences, which refer to equilibrium, can be determined via
non-equilibrium processes \cite{Jarzynski1997,Crooks1999}.

Typically, free-energy differences are beyond the scope of analytic
computations and one needs to measure them experimentally or
compute them numerically. Highly efficient methods have been
developed in order to estimate free-energy differences precisely,
including thermodynamic integration \cite{Kirkwood1935,Gelman1998},
free-energy perturbation \cite{Zwanzig1954},
umbrella sampling \cite{Torrie1977,Chen1997,Oberhofer2008},
adiabatic switching \cite{Watanabe1990},
dynamic methods \cite{Sun2003,Ytreberg2004,Jarzynski2006},
asymptotics of work distributions \cite{vonEgan-Krieger2008},
optimal protocols \cite{Then2008}, targeted and escorted free-energy
perturbation \cite{Meng2002,Jarzynski2002,Oberhofer2007,
Vaikuntanathan2008,Hahn2009}.

A powerful \cite{Meng1996,Kong2003,Shirts2008} and frequently
\cite{Ceperley1995,Frenkel2002,Collin2005} used method for free-energy
determination is two-sided estimation, i.e.\ Bennett's acceptance ratio
method \cite{Bennett1976}, which employs a sample of work values of
a driven nonequilibrium process together with a sample of work
values of the time-reversed process \cite{Crooks2000}.

The performance of two-sided free-energy estimation depends on the ratio
\begin{align}\label{rdef}
  r = \frac{n_1}{n_0}
\end{align}
of the number of forward and reverse work values used. Think
of an experimenter who wishes to estimate the free-energy
difference with Bennett's acceptance ratio method and has the possibility
to generate forward as well as reverse work values.
The capabilities of the experiment give rise to an obvious question:
if the total amount of draws is intended to be $N=n_0+n_1$, which is
the optimal choice of partitioning $N$ into the numbers $n_0$ of forward
and $n_1$ of reverse work values, or equivalently, what is the optimal choice
$r_{o}$ of the ratio $r$? The problem is to determine the value
of $r$ that minimizes the (asymptotic) mean square error of Bennett's
estimator when $N=n_0+n_1$ is held constant.

While known since Bennett \cite{Bennett1976},
the optimal ratio is underutilized in the
literature. Bennett himself proposed to use a suboptimal
equal time strategy, instead, because his estimator for
the optimal ratio converges too slowly in order to be practicable.
Even questions as fundamental as the existence and uniqueness are
unanswered in the literature. Moreover, it is
not always clear a priori whether two-sided
free-energy estimation is better than one-sided exponential
work averaging. For instance, Shirts et al.\ have
presented a physical example where it is optimal to draw
work values from only one direction \cite{Shirts2005}.

The paper is organized as follows:
in Secs.~\ref{sec:2} and
\ref{sec:3} we rederive two-sided free-energy estimation and the
optimal ratio. We also remind that two-sided estimation
comprises one-sided exponential work averaging as limiting cases for
$\ln r\to\pm\infty$, a result that is also
true for the mean square errors of the corresponding estimators.

The central result is stated in Sec.~\ref{sec:4}:
the asymptotic mean square error of two-sided estimation is convex
in the fraction $\frac{n_0}{N}$ of forward
work values used. This fundamental characteristic
immediately implies that the optimal ratio $r_o$ exists and is unique.
Moreover, it explains the generic superiority of two-sided estimation if
compared with one-sided, as found in many applications.

To overcome the slow convergence of Bennett's
estimator of the optimal ratio, which is based on estimating second moments,
in Sec.~\ref{sec:5} we transform
the problem into another form such that the corresponding estimator
is entirely based on first moments, which enhances the convergence
enormously.

As an application, in Sec.~\ref{sec:7} we present a dynamic strategy of
sampling forward and reverse work values that maximizes the efficiency
of two-sided free-energy estimation.

\section{Two-sided free-energy estimation}\label{sec:2}

\textit{Given} a pair of samples of $n_0$ forward and $n_1$ reverse work
values drawn from the probability densities
$\pf(w)$ and $\pr(w)$ of forward and reverse work values and
provided the latter are related to each other
via the fluctuation theorem \cite{Crooks1999},
\begin{align}\label{fth}
  \frac{\pf(w)}{\pr(w)} = e^{w-\df},
\end{align}
Bennett's acceptance ratio method
\cite{Bennett1976,Meng1996,Crooks2000,Shirts2003}
is known to give the optimal
estimate of the free-energy difference $\df$ in
the limit of large sample sizes.
Throughout the paper, $\df=\Delta F/kT$
and $w=W/kT$ are understood to be measured
in units of the thermal energy $kT$. The normalized probability
densities $\pf(w)$ and $\pr(w)$ are assumed to
have the same support $\Omega$, and we choose the
following sign convention: $\pf(w):=p_{\text{forward}}(+w)$ and
$\pr(w):=p_{\text{reverse}}(-w)$.

\begin{figure}
  \includegraphics{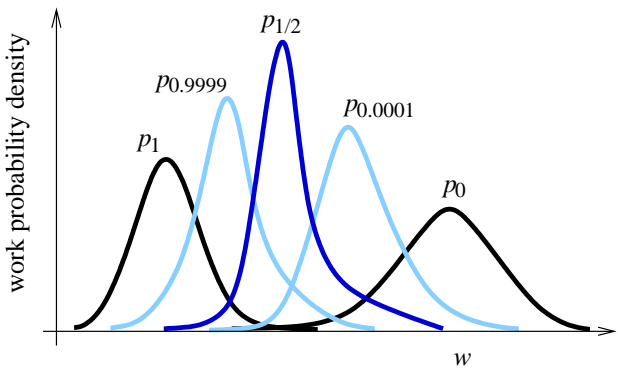}
  \caption{\label{fig:1} (Color online) The overlap density $\pa(w)$ bridges
    the densities $\pf(w)$ and $\pr(w)$ of forward and reverse work values,
    respectively. $\al$ is the fraction $\frac{n_0}{n_0+n_1}$ of forward
    work values, here schematically shown for $\al=0.0001$, $\al=0.5$,
    and $\al=0.9999$\,.
    The accuracy of two-sided free-energy estimates depends
    on how good $\pa(w)$ is sampled when drawing from $\pf(w)$ and
    $\pr(w)$.}
\end{figure}

Now define a normalized density $\pa(w)$ with
\begin{align}\label{pa}
  \pa(w) = \frac{1}{\Ua} \frac{\pf(w)\pr(w)}{\al\pf(w)+\be\pr(w)},
\end{align}
$w\in\Omega$, where  $\al\in [0,1]$ is a real number and
\begin{align}
  \al+\be=1.
\end{align}
The normalization constant $\Ua$ is given by
\begin{align} \label{Udef}
  \Ua = \int\limits_{\Omega} \frac{\pf\pr}{\al\pf+\be\pr} dw.
\end{align}
The density $\pa(w)$ is a normalized harmonic mean of $\pf$ and $\pr$,
$\frac{\pf\pr}{\al\pf+\be\pr}=
\left[\al\frac{1}{\pr}+\be\frac{1}{\pf}\right]^{-1}$,
and thus bridges between $\pf$ and $\pr$, see Fig.~\ref{fig:1}.
In the limit $\al\rightarrow 0$, $\pa(w)$ converges to the forward work
density $\pf(w)$, and conversely for $\al\rightarrow 1$ it converges
to the reverse density $\pr(w)$.
As a consequence of the inequality of the harmonic and arithmetic mean,
$\left[\al\frac{1}{\pr}+\be\frac{1}{\pf}\right]^{-1} \leq \al\pr+\be\pf$,
$\Ua$ is bounded from above by unity,
\begin{align}
  \Ua \leq 1
\end{align}
$\forall \al\in[0,1]$. Except for $\al=0$ and $\al=1$,
the equality holds if and only if $\pf\equiv\pr$. Using the
fluctuation theorem \gl{fth}, $\Ua$ can be written as an
average in $\pf$ and $\pr$,
\begin{align}\label{Uforms}
  \Ua = \meanr{ \frac{1}{\al + \be e^{-w+\df}} }
  = \meanf{ \frac{1}{\be + \al e^{w-\df}} },
\end{align}
where the angular brackets with subscript $\gamma\in[0,1]$ denote
an ensemble average
with respect to $p_\gamma$, i.e.\ 
\begin{align}
  \mean{g}{\gamma} = \int\limits_{\Omega} g(w) p_\gamma(w) dw
\end{align}
for an arbitrary function $g(w)$.

In setting $\al=1$, Eq.~\gl{Uforms} reduces to the nonequilibrium
work relation \cite{Jarzynski1997}
\begin{align}\label{noneqf}
  1 = \meanf{ e^{-w+\df } }
\end{align}
in the \textit{forward} direction, and conversely with $\al=0$ we obtain
the nonequilibrium work relation in the \textit{reverse} direction,
\begin{align}\label{noneqr}
  1 = \meanr{ e^{w-\df} }.
\end{align}
The last two relations can, of course, be obtained more directly from
the fluctuation theorem \gl{fth}. An important application of these relations
is the \textit{one-sided} free-energy estimation:
Given a sample $\{w^0_1\ldots w^0_N\}$ of $N$ forward
work values drawn from $\pf$, Eq.~\gl{noneqf} is commonly used to define
the \textit{forward} estimate $\dfhat_0$ of $\df$ with
\begin{align}\label{dfhatf}
  \dfhat_0 = - \ln \frac{1}{N} \sum\limits_{k=1}^N e^{-w^0_k}.
\end{align}
Conversely, given a sample $\{w^1_1\ldots w^1_N\}$ of $N$ reverse work values
drawn from $\pr$, Eq.~\gl{noneqr} suggests the definition of the
\textit{reverse} estimate $\dfhat_1$ of $\df$,
\begin{align}\label{dfhatr}
\dfhat_1 = \ln \frac{1}{N} \sum\limits_{l=1}^N e^{w^1_l}.
\end{align}

If we have drawn both, a sample of $n_0$ forward \textit{and} a sample of
$n_1$ reverse work values, then Eq.~\gl{Uforms} can
serve us to define a two-sided estimate $\dfhat$ of $\df$ by replacing
the ensemble averages
with sample averages:
\begin{align}\label{benest}
  \frac{1}{n_1} \sum\limits_{l=1}^{n_1}\frac{1}{\al + \be e^{-w^1_l+\dfhat}}
  = \frac{1}{n_0} \sum\limits_{k=1}^{n_0} \frac{1}{\be + \al e^{w^0_k-\dfhat}}.
\end{align}
$\dfhat$ is understood to be the unique root of Eq.~\gl{benest}, which exists
for any $\al\in[0,1]$. Different values of $\al$ result in different estimates
for $\df$. Choosing
\begin{align}\label{alphabeta}
\al = \frac{n_0}{N}, \quad \be = \frac{n_1}{N},
\end{align}
$N=n_0+n_1$, the estimate \gl{benest}
coincides with Bennett's optimal estimate, which defines the
two-sided estimate with least asymptotic mean square error for a given value
$\al=\frac{n_0}{N}$, or equivalently,
\textit{for a given ratio $r=\frac{\be}{\al}=\frac{n_1}{n_0}$}
\cite{Bennett1976,Meng1996}.
We denote the optimal two-sided estimate, i.e.\ the solution of
Eq.~\gl{benest} under the constraint
\gl{alphabeta}, by $\dfhatopt$ and simply refer to
it as the two-sided estimate.
Note that the optimal estimator can be
written in the familiar form
\begin{align}\label{benestfam}
  \sum\limits_{l=1}^{n_1} \frac{1}{1 + e^{-w^1_l+\dfhat+\ln \frac{n_1}{n_0}}}
  = \sum\limits_{k=1}^{n_0} \frac{1}{1 + e^{w^0_k-\dfhat-\ln \frac{n_1}{n_0}}}.
\end{align}

In the limit $\al=\frac{n_0}{N}\rightarrow 1$ the two-sided estimate
reduces to the one-sided forward estimate \gl{dfhatf},
$\dfhatopt \overset{\al\to1}{\longrightarrow} \dfhat_0$,
and conversely
$\dfhatopt \overset{\al\to 0}{\longrightarrow} \dfhat_1$.
Thus the one-sided estimates are the optimal estimates
if we have given draws from only one of the densities $\pf$ or $\pr$.

A characteristic quantity to express the performance of the estimate
$\dfhatopt$ is the mean square error,
\begin{align}
  \la \left(\dfhatopt-\df \right)^2 \ra,
\end{align}
which depends on the total sample size $N=n_0+n_1$ and the fraction
$\al=\frac{n_0}{N}$.
Here, the average is understood to be an ensemble average
in the value distribution  of the estimate $\dfhatopt$ for fixed $N$
and $\al$.
In the limit of large $n_0$ and $n_1$, the asymptotic mean square error $X$
(which then equals the variance) can be
written \cite{Bennett1976,Meng1996}
\begin{align}\label{mse}
  X(N,\al) = \frac{1}{N}\frac{1}{\al\be} \left( \frac{1}{\Ua}-1\right).
\end{align}
Provided the r.h.s.\ of Eq.~\gl{mse} exists, which is guaranteed for any
$\al\in (0,1)$, the $N$-dependence of $X$ is simply given by the usual
$\frac{1}{N}$-factor, whereas the $\al$-dependence is
determined by the function $\Ua$ given in Eq.~\gl{Udef}.
Note that if a two-sided estimate $\dfhatopt$ is calculated, then
essentially the normalizing
constant $\Ua$ is estimated from two sides, $0$ and $1$, cf.\ 
Eqs.~\gl{Uforms} and \gl{benest}.
With an estimate $\dfhatopt$ we therefore always have an estimate of
the mean square error at hand.
However, the reliability of the latter naturally depends on the degree
of convergence of the estimate
$\dfhatopt$. The convergence of the two-sided estimate can be checked with the
convergence measure introduced in Ref.~\cite{Hahn2009}.

In the limits $\al=\frac{n_0}{N}\to1$ and $\al\to0$, respectively,
the asymptotic mean square error $X$
of the two-sided estimator converges to the asymptotic
mean square error of the appropriate one-sided
estimator \cite{Gore2003},
\begin{align}\label{Mf}
  \lim\limits_{\al\rightarrow 1}{X(N,\al)}
  = \frac{1}{N} \varf{\frac{\pr}{\pf}}=\frac{1}{N} \varf{e^{-w+\df}}
\end{align}
and
\begin{align}\label{Mr}
  \lim\limits_{\al\rightarrow 0}{X(N,\al)}
  = \frac{1}{N} \varr{\frac{\pf}{\pr}}=\frac{1}{N} \varr{e^{w-\df}},
\end{align}
where
$\var_\gamma$ denotes the variance operator with respect to the density
$p_\gamma$, i.e.\ 
\begin{align}
  \var_\gamma(g) = \mean{ \left(g-\mean{g}{\gamma} \right)^2}{ \gamma}
\end{align}
for an arbitrary function $g(w)$ and $\gamma\in[0,1]$.

\section{The optimal ratio}\label{sec:3}

Now we focus on the question raised in the introduction:
Which value $\al_o$ of $\al$ in the
range $[0,1]$ minimizes the mean square error
\gl{mse} when the total sample size, $N=n_0+n_1$, is held fixed?

Let $M$ be the rescaled asymptotic mean square error given by
\begin{align}\label{Mdef}
  M(\al) = N\cdot X(N,\al),
\end{align}
which is a function of $\al$ only. Assuming $\al_o\in(0,1)$, a
necessary condition
for a minimum of $M$ is that the derivative $M'(\al)=\frac{dM}{d\al}$
of $M$ vanishes
at $\al_o$.
Before calculating $M'$ explicitly, it is beneficial to rewrite $M$ by
using the identity
\begin{multline}\label{Uforms2}
  \Ua = \int\limits_{\Omega}
  \frac{\pf\pr\left(\al\pf+\be\pr\right)}{\left(\al\pf+\be\pr\right)^2}dw \\
  = \al \meanr{  \frac{\pf^2}{\left(\al\pf+\be\pr\right)^2} }
  + \be \meanf{ \frac{\pr^2}{\left(\al\pf+\be\pr\right)^2} }.
\end{multline}
Subtracting $(\al+\be)\Ua^2=\Ua^2$ from Eq.~\gl{Uforms2} and recalling the
definition \gl{pa} of $\pa$, one obtains
\begin{align}\label{varid}
  \Ua\left( 1-\Ua \right)
  = \left[\al \theta_1(\al) + \be \theta_0(\al)\right] \Ua^2,
\end{align}
where the functions $\theta_i$ are defined as
\begin{align}\label{thetadef}
  &\theta_1(\al) = \varr{\frac{\pa}{\pr}}
  = \frac{1}{\Ua^2} \varr{ \frac{1}{ \al+\be e^{-w+\df} } },\nonumber \\
  &\theta_0(\al) = \varf{\frac{\pa}{\pf}}
  = \frac{1}{\Ua^2} \varf{ \frac{1}{ \be+\al e^{w-\df} } }.
\end{align}
$\theta_0$ and $\theta_1$ describe the relative fluctuations of the
quantities that are averaged in the two-sided estimation of $\df$,
cf.\ Eq.~\gl{benest}.

With the use of formula \gl{varid}, $M$ can be written
\begin{align}\label{M1}
  M(\al) = \frac{\theta_0(\al)}{\al} + \frac{\theta_1(\al)}{\be}
\end{align}
and the derivative yields
\begin{align}\label{M'}
  M'(\al) = \frac{\theta_1(\al)}{\be^2}-\frac{\theta_0(\al)}{\al^2}
  + \frac{\be\theta_0'(\al)+\al\theta_1'(\al) }{\al\be}.
\end{align}
The derivatives of the $\theta$-functions involve
the first two derivatives of $\Ua$, which will thus be computed first:
\begin{align}\label{U'}
  \Ua' := \frac{d}{d\al}\Ua = \int\limits_{\Omega}
  \frac{\pf\pr\left( \pr-\pf \right)}{\left( \al\pf + \be\pr \right)^2} dw
\end{align}
and
\begin{align}\label{U''}
  \Ua'' := \frac{d^2}{d\al^2}\Ua = 2\int\limits_{\Omega}
  \frac{\pf\pr\left( \pr-\pf \right)^2}{\left( \al\pf + \be\pr \right)^3} dw.
\end{align}
From this equation it is clear that $\Ua$ is convex in $\al$, $\Ua'' \geq 0$,
with a unique minimum in $(0,1)$ (as $U_0=U_1=1$). We can rewrite the
$\theta$-functions
with $\Ua$ and $\Ua'$ as follows:
\begin{align}\label{thetaforms}
  \theta_1(\al) = \frac{\Ua-\be\Ua'}{\Ua^2} -1, \nonumber \\
  \theta_0(\al) = \frac{\Ua+\al\Ua'}{\Ua^2} -1.
\end{align}
Differentiating these expressions gives
\begin{align}\label{theta'}
  &\theta_1'(\al) = -\frac{\be}{\Ua^3}\left(\Ua''\Ua - 2\Ua'^2 \right),
  \nonumber \\
  &\theta_0'(\al) = \frac{\al}{\Ua^3}\left(\Ua''\Ua - 2\Ua'^2 \right).
\end{align}
$\theta_0$ and $\theta_1$ are monotonically increasing and decreasing,
respectively. This immediately follows from writing the term
occurring in the brackets of Eqs.~\gl{theta'} as
a variance in the density $\pa$,
\begin{align}
  \Ua''\Ua - 2\Ua'^2 = 2 \varal{\frac{\pr-\pf}{\al\pf+\be\pr}} \Ua^2,
\end{align}
which is thus positive.

As a consequence of Eq.~\gl{theta'}, the relation
\begin{align}\label{theta'rel}
  \be\theta_0'(\al)+\al\theta_1'(\al)=0 \quad \forall \al\in[0,1]
\end{align}
holds and $M'$ reduces to
\begin{align}\label{M'2}
  M'(\al) =  \frac{\theta_1(\al)}{\be^2}-\frac{\theta_0(\al)}{\al^2} .
\end{align}
The derivatives of the $\theta$-functions do not
contribute to $M'$ due to the fact that the special form of
the two-sided estimator
\gl{benest} originates from minimizing the asymptotic mean square
error, cf.~\cite{Bennett1976}. The necessary condition for a local
minimum of $M$
at $\al_o$, $M'(\al_o)=0$, now reads
\begin{align}\label{mincond}
  \frac{\be_o^2}{\al_o^2} = \frac{\theta_1(\al_o)}{\theta_0(\al_o)},
\end{align}
where $\be_o=1-\al_o$ is introduced. Using Eqs.~\gl{thetadef} and \gl{fth},
the condition \gl{mincond} results in
\begin{align}\label{mincond2}
  \varr{\frac{1}{1+e^{-w+\df+\ln r_o}}}
  = \varf{\frac{1}{1+e^{w-\df-\ln r_o}}}.
\end{align}
This means, the
optimal ratio $r_o$ is such that the variances of the random functions
which are averaged in the two-sided estimation \gl{benestfam} are equal.
However, the existence of a solution of $M'(\al)=0$ is not guaranteed in
general.

Writing Eq.~\gl{mincond2} in the form
\begin{align} \label{mincond3}
  \varr{\frac{p_1-p_0}{\al p_0+\be p_1}}
  = \varf{\frac{p_1-p_0}{\al p_0+\be p_1}}
\end{align}
prevents the equation from becoming a tautology.

\section{Convexity of the mean square error}\label{sec:4}

\begin{thm}
  The asymptotic mean square error $M(\al)$ is convex in $\al$.
\end{thm}

In order to prove the convexity, we introduce the operator
$\Gamma_\al\left(f\right)$ which is defined for an arbitrary function
$f(w)$ by
\begin{align}
  \Gamma_\al\left( f\right)= \be \varf{f} + \al \varr{f} - \Ua \varal{f}.
\end{align}

\begin{lem}
  $\Gamma_\al$ is positive semidefinite, i.e.\ 
  \begin{align}
    \Gamma_\al(f) \geq 0 \quad \forall f(w).
  \end{align}
  For $\al\in(0,1)$ and $f(w)\neq \textit{const.}$, the equality holds
  if and only if $\pf\equiv \pr$.
\end{lem}

\begin{proof}[Proof of the Lemma.]
  Let $\deltf_\gamma = f(w) - \la f \ra_\gamma$, $\gamma\in[0,1]$.
  Then
  \begin{multline}
    \Gamma_\al\left(f\right)= \int\limits_{\Omega} \left(\be \deltf_0^2 \pf
    +\al \deltf_1^2 \pr - \deltf_\al^2 \frac{\pf\pr}{\al\pf+\be\pr} \right)dw\\
    = \int\limits_{\Omega}
    \frac{ \left(\be \deltf_0^2 \pf +\al \deltf_1^2 \pr\right)
      \left(\al\pf +\be\pr \right) - \deltf_\al^2 \pf\pr }{\al\pf+\be\pr}dw \\
    = \al\be \int\limits_{\Omega}
    \frac{\left(\deltf_1\pr-\deltf_0 \pf \right)^2}{\al\pf+\be\pr} dw \\
    + \Ua \left(\be\la f \ra_0 + \al \la f \ra_1 - \la f \ra_\al\right)^2,
  \end{multline}
  which is clearly positive. Provided $f\neq\textit{const.}$ and
  $\al\neq 0,1$, the integrand
  in the last line is zero $\forall w$ if and only if $\pf\equiv\pr$.
  This completes the proof of the Lemma.
\end{proof}

\begin{proof}[Proof of the Theorem.]
  Consulting Eqs.~\gl{M'2} and \gl{theta'rel},
  the second derivative of $M$ reads
  \begin{align}\label{M''}
    M''(\al) =  2 \left( \frac{\theta_1(\al)}{\be^3}
    + \frac{\theta_0(\al)}{\al^3}\right)
    - \frac{1}{\al^2\be} \theta_0'(\al).
  \end{align}
  Expressing $\pf=p-\be d$ and $\pr=p+\al d$ in center- and relative
  ``coordinates'' $p=\al\pf+\be\pr$ and $d=\pr-\pf$, respectively, gives
  \begin{align}
    &\theta_1(\al) = \frac{1}{\Ua^2}\varr{\frac{\pf}{p}}
    = \frac{\be^2}{\Ua^2}\varr{\frac{d}{p}},\nonumber \\
    &\theta_0(\al) = \frac{1}{\Ua^2}\varf{\frac{\pr}{p}}
    = \frac{\al^2}{\Ua^2}\varf{\frac{d}{p}},\nonumber \\
    &\theta_0'(\al) = \frac{2\al}{\Ua}\varal{\frac{d}{p}}.
  \end{align}
  Therefore, $\frac{1}{2}\al\be \Ua^2M''= \Gamma_\al\big( \frac{d}{p}\big)$,
  which is positive according to the Lemma.
\end{proof}

The convexity of the mean square error is a fundamental characteristic
of Bennett's acceptance ratio method. This characteristic allows us to state a
simple criterion for the existence of a \textit{local} minimum of the mean
square error in terms of its derivatives at the boundaries. Namely, if
\begin{align} \label{Mp0}
  M'(0) = \varr{e^{w-\df}} - \varf{e^{w-\df}}
\end{align}
is negative and
\begin{align} \label{Mp1}
  M'(1) = \varr{e^{-w+\df}} - \varf{e^{-w+\df}}
\end{align}
is positive there exists a local minimum of $M(\al)$ for $\al\in(0,1)$.
Otherwise, no local minimum exists and the global minimum
is found on the boundaries of $\al$: if $M'(0)> 0$, the global minimum
is found for $\al=0$, thus it is optimal to measure work
values in the reverse direction only and to use the one-sided reverse
estimator \gl{dfhatr}. Else, if $M'(1)<0$, the global minimum
is found for $\al=1$, implying the one-sided forward estimator \gl{dfhatf}
to be optimal.

In addition, the convexity of the mean square error proves the
existence and uniqueness of the optimal ratio, since a convex
function has a global minimum on a closed interval.

\begin{cor}
  If a solution of $M'(\al)=0$ exists, it is unique and $M(\al)$ attains its
  global minimum ($\al\in[0,1]$) there.
\end{cor}

\section{Estimating the optimal ratio with first moments}\label{sec:5}

In situations of practical interest the optimal ratio is not available
\textit{a priori}. Thus, we are going to estimate the optimal ratio.
There exist estimators of the optimal ratio since Bennett. In addition
we have just proven that the optimal ratio exists and is unique.
However there is still one obstacle to overcome. Yet, all expressions
for estimating the optimal ratio are based on second moments,
see e.g.\ Eq.~\gl{mincond2}.
Due to convergence issues, it is not practicable to base any estimator
on expressions that involve second moments. The estimator would converge
far too slowly. For this reason, we transform
the problem into a form that employs first moments, only.

Assume we have given $n_0$ and $n_1$ work values in forward
and reverse direction, respectively, and want to estimate $U_a$,
with $0\le a\le1$.
According to Eq.~\gl{Uforms} we can estimate the overlap measure $U_a$
by using draws from the forward direction,
\begin{align} \label{eU0}
  \Uhat_a^{(0)} =
  \frac{1}{n_0} \sum\limits_{k=1}^{n_0} \frac{1}{b + a e^{w^0_k-\dfhat}},
\end{align}
where $b$ equals $1-a$ and for $\dfhat$ the best available estimate of
$\df$ is inserted, i.e.\ the two-sided estimate based on the $n_0+n_1$
work values.
Similarly, we can estimate the overlap measure by using draws from the
reverse direction,
\begin{align} \label{eU1}
  \Uhat_a^{(1)} =
  \frac{1}{n_1} \sum\limits_{l=1}^{n_1} \frac{1}{a + b e^{-w^1_l+\dfhat}}.
\end{align}
Since in general draws from both directions are available, it is reasonable
to take an arithmetic mean of both estimates
\begin{align} \label{eU}
  \Uhat_a = a \Uhat_a^{(1)} + b \Uhat_a^{(0)},
\end{align}
where the weighting is chosen such that the better estimate, $\Uhat_a^{(0)}$
or $\Uhat_a^{(1)}$, contributes stronger: with increasing $a$ the estimate
$\Uhat_a^{(1)}$ becomes more reliable, as $U_a$ is the normalizing constant of
the bridging density $p_a$, Eq.~\gl{pa}, and
$p_a \xrightarrow{a\rightarrow 1} \pr $; and conversely for decreasing $a$.

From the estimate of the overlap measure we can estimate the
rescaled mean square error by
\begin{align} \label{Mhat}
  \Mhat(a) = \frac{1}{ab} \left( \frac{1}{\Uhat_a} - 1 \right)
\end{align}
for all $a\in(0,1)$, a result that is entirely based on first moments.
The infimum of $\Mhat(a)$ finally results in an estimate $\alhat_o$ of the
optimal choice $\al_o$ of $\frac{n_0}{N}$,
\begin{align} \label{orfm}
  \alhat_o \quad :\Leftrightarrow \quad \Mhat(\alhat_o) = \inf_{a} \Mhat(a).
\end{align}
When searching for the infimum, we also take
\begin{align} \label{Mat0}
  \Mhat(0) &= \frac{1}{n_0} \sum_{k=1}^{n_0} e^{w_k^{(0)}-\dfhat}
  - \frac{1}{n_1} \sum_{l=1}^{n_1} e^{w_l^{(1)}-\dfhat},
  \\
  \Mhat(1) &= \frac{1}{n_1} \sum_{l=1}^{n_1} e^{-w_l^{(1)}+\dfhat}
  - \frac{1}{n_0} \sum_{k=1}^{n_0} e^{-w_k^{(0)}+\dfhat} \nonumber
\end{align}
into account which follow from a series expansion of Eq.~\gl{Mhat} in
$a$ at $a=0$ and $a=1$, respectively.

\section{Incorporating costs}\label{sec:6}

The costs of measuring a work value in forward direction may differ
from the costs of measuring a work value in reverse direction.
The influence of costs on the optimal ratio of sample sizes is
investigated here.

Different costs can be due to a direction dependent
effort of experimental or computational measurement
of work (unfolding a RNA may be much easier than folding
it). We assume the work values to be uncorrelated, which
is essential for the validity of the theory presented in
this paper. Thus, a source of nonequal costs, which arises
especially when work values are obtained via computer
simulations, is the difference in the strength of
correlations of consecutive Monte-Carlo steps in forward and
reverse direction. To achieve uncorrelated draws, the
``correlation-lengths'' or ``correlation-times''
have to be determined within the simulation, too.
However, this is advisable in any case of two-sided
estimation, independent of the sampling strategy.  

Let $c_0$ and $c_1$ be the costs of drawing a single forward and
reverse work value, respectively.
Our goal is to minimize the mean square error
$X=\frac{1}{N}M$ while keeping the total costs $c=n_0c_0+n_1c_1$
constant. Keeping $c$ constant results in
\begin{align} \label{Ncc}
  N(c,\al) = \frac{c}{\al c_0 + \be c_1}
\end{align}
which in turn yields
\begin{align} \label{Xca}
  X(c,\al) = \frac{1}{N(c,\al)} M(\al).
\end{align}
If a local minimum exists, it results from
$\frac{\partial}{\partial\al}X(c,\al)=0$ which leads to
\begin{align}\label{cmincond}
  \frac{\be_o^2}{\al_o^2} = \frac{c_0\theta_1(\al_o)}{c_1\theta_0(\al_o)},
\end{align}
a result Bennett was already aware of \cite{Bennett1976}. However, based
on second moments, it was not possible to estimate the optimal ratio $r_o$
accurately and reliably. Hence, Bennett proposed to use a suboptimal
\textit{equal time strategy} or \textit{equal cost strategy},
which spends an equal amount of expenses to both directions, i.e.\ 
$n_0c_0=n_1c_1=\frac{c}{2}$ or
\begin{align} \label{ecs}
  \frac{\be_{ec}}{\al_{ec}} = \frac{c_0}{c_1},
\end{align}
where $\al_{ec}=1-\be_{ec}$ is the equal cost choice for $\al=\frac{n_0}{N}$.
This choice is motivated by the following result
\begin{align} \label{Xec}
  X(c,\al) \ge \frac{1}{2} X(c,\al_{ec}) \quad \forall\al\in[0,1]
\end{align}
which states that the asymptotic
mean square error of the equal cost strategy is
at most sub-optimal by a factor of $2$ \cite{Bennett1976}.
Note however that the equal cost strategy can be far more sub-optimal
if the asymptotic limit of large sample sizes is not reached.

Since we can base the estimator for the optimal ratio $r_o$
on first moments, see Sec.~\ref{sec:5}, we propose a
\textit{dynamic strategy}
that performs better than the equal cost strategy.
The infimum of
\begin{align} \label{Xhat}
  \Xhat(c,a) = \frac{a c_0 + b c_1}{c} \Mhat(a)
\end{align}
results in the estimate $\alhat_o$ of the optimal choice
$\al_o$ of $\frac{n_0}{N}$,
\begin{align} \label{orfmc}
  \alhat_o \quad :\Leftrightarrow \quad \Xhat(c,\alhat_o) = \inf_{a} \Xhat(c,a).
\end{align}

We remark that opposed to $M(\al)$, $X(c,\al)$ is not necessarily convex.
However, a global minimum clearly exists and can be estimated.

\section{A dynamic sampling strategy}\label{sec:7}

Suppose we want to estimate the free-energy difference with
the acceptance ratio method, but have a limit on the total amount
of expenses $c$ that can be spend for measurements of work.
In order to maximize the
efficiency, the measurements are
to be performed such that $\frac{n_0}{N}$ finally equals the
optimal fraction $\al_o$ of forward measurements.

The dynamic strategy is as follows:
\begin{enumerate}
\item
  In absence of preknowledge on $\al_o$, we start with
  Bennett's equal cost strategy \gl{ecs} as an initial guess of $\al_o$.
\item
  After drawing a small number of work values we make preliminary
  estimates of the free-energy difference, the mean square error,
  and the optimal fraction $\al_o$.
\item
  Depending on whether the estimated rescaled mean square error
  $\Mhat(a)$ is convex, which is
  a necessary condition for convergence, our algorithm updates
  the estimate $\alhat_o$ of $\al_o$.
\item
  Further work values are drawn such that $\frac{n_0}{N}$
  dynamically follows $\alhat_o$, while $\alhat_o$ is updated
  repeatedly.
\end{enumerate}
There is no need to update $\alhat_o$ after each individual draw.
Splitting the total costs into a sequence $0<c^{(1)}<\ldots<c^{(p)}=c$,
not necessarily equidistant, we can predefine when and how often
an update in $\alhat_o$ is made. Namely, this is done whenever the
actually spent costs reach the next value $c^{(\nu)}$ of the sequence.

The dynamic strategy can be cast into an algorithm.
\begin{alg} \rm
  Set the initial values $n_0^{(0)}=n_1^{(0)}=0$,
  $\alhat_o^{(1)}=\al_{ec}$.
  In the $\nu$-th step of the iteration, $\nu=1,\ldots,p$, determine
  \begin{align} \label{en0n1}
    n_0^{(\nu)} &= \lfloor \alhat_o^{(\nu)} N^{(\nu)} \rfloor
    \\
    n_1^{(\nu)} &= \lfloor \behat_o^{(\nu)} N^{(\nu)} \rfloor
    \nonumber
  \end{align}
  with
  \begin{align} \label{eN}
    N^{(\nu)} = \frac{c^{(\nu)}}{\alhat_o^{(\nu)}c_0 + \behat_o^{(\nu)}c_1},
  \end{align}
  where $\lfloor\ \rfloor$ means rounding to the next lower integer.
  Then, $\Delta n_0^{(\nu)}=n_0^{(\nu)}-n_0^{(\nu-1)}$ additional forward
  and $\Delta n_1^{(\nu)}=n_1^{(\nu)}-n_1^{(\nu-1)}$ additional reverse
  work values are drawn. Using the entire present samples, an estimate
  $\dfhat^{(\nu)}$ of $\df$ is calculated according to Eq.~\gl{benest}.
  With the free-energy estimate at hand, $\Mhat^{(\nu)}(a)$ is calculated
  for all values of $a\in[0,1]$ via Eqs.~\gl{eU0}--\gl{Mhat} and \gl{Mat0},
  discretized, say in steps $\Delta a=0.01$\,.
  If $\Mhat^{(\nu)}(a)$ is convex, we update the recent estimate
  $\alhat_o^{(\nu)}$ of $\al_o$ to
  $\alhat_o^{(\nu+1)}$ via Eqs.~\gl{Xhat} and \gl{orfmc}.
  Otherwise, if $\Mhat^{(\nu)}(a)$ is not convex, the corresponding
  estimate of $\al_o$ is not yet reliable and we keep the recent
  value, $\alhat_o^{(\nu+1)}=\alhat_o^{(\nu)}$.
  Increasing $\nu$ by one, we iteratively continue with Eq.~\gl{en0n1}
  until we finally obtain $\dfhat^{(p)}$ which is the optimal estimate of
  the free-energy difference after having spend all costs $c$.

  Note that an update in $\alhat_o^{(\nu)}$ may result in negative
  values of $\Delta n_0^{(\nu)}$ or $\Delta n_1^{(\nu)}$. Should
  $\Delta n_0^{(\nu)}$ happen to be negative, we set
  $n_0^{(\nu)}=n_0^{(\nu-1)}$ and
  \begin{align} \label{n0neg}
    n_1^{(\nu)}
    = \left\lfloor \frac{c^{(\nu)} - c_0 n_0^{(\nu-1)}}{c_1} \right\rfloor.
  \end{align}
  We proceed analogously, if $\Delta n_1^{(\nu)}$ happens to be negative.
\end{alg}

The optimal fraction $\al_o$ depends on the cost ratio $c_1/c_0$, i.e.\ 
the algorithm needs to know the costs $c_0$ and $c_1$. However, the
costs are not always known in advance and may also vary over time.
Think of a long time experiment which is subject to currency changes,
inflation, terms of trade, innovations, and so on. Of advantage is that
the dynamic sampling strategy is capable of incorporating varying costs.
In each iteration step of the algorithm one just inserts the actual costs.
If desired, the breakpoints $c^{(\nu)}$ may also be adapted to the
actual costs. Should the costs initially be unknown (e.g.\ the
``correlation-length'' of a Monte-Carlo simulation needs to be
determined within the simulation first) one may use any reasonable
guess until the costs are known.

\section{An example}\label{sec:8}

For illustration of results we choose exponential work distributions
\begin{align}\label{expdist}
  p_i(w) = \frac{1}{\mu_i} e^{-\frac{w}{\mu_i}},
  \quad w \in \Omega = \mathds{R}^+,
\end{align}
$\mu_i>0$, $i=0,1$. According to the fluctuation theorem
\gl{fth} we have $\mu_1=\frac{\mu_0}{1+\mu_0}$ and $\df=\ln(1+\mu_0)$.

\begin{figure}
  \includegraphics{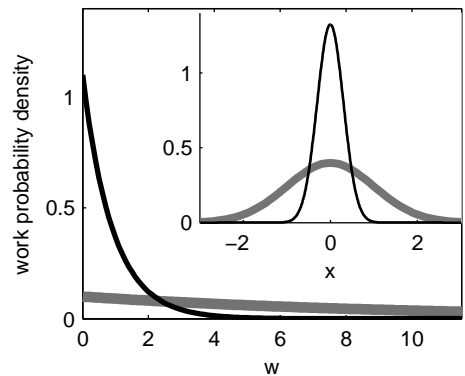}
\caption{\label{fig:2} The main figure displays
    the exponential work densities $\pf$ (thick line) and $\pr$ (thin line)
    for the choice of $\mu_0=10$ and, according to the fluctuation theorem,
    $\mu_1=10/11$. The inset displays the corresponding Boltzmann
    distributions $\rho_0(x,y)$ (thick) and $\rho_1(x,y)$ (thin)
    both for $y=0$. Here, $\omega_0$ is set equal to $1$ arbitrarily, hence
    $\omega_1^2=(1+\mu_0)\omega_0^2 =11$. The free-energy difference
    is $\df=\ln(1+\mu_0)=\ln(\omega_1^2/\omega_0^2)\approx 2.38$\,.}
\end{figure}

Exponential work densities arise in a natural way
in the context of a two-dimensional harmonic oscillator with Boltzmann
distribution $\rho(x,y)=e^{-\frac{1}{2}\omega^2(x^2+y^2)}/Z$,
where $Z=2\pi/\omega^2$ is a normalizing constant (partition function)
and $(x,y)\in\mathds{R}^2$ \cite{Shirts2005}.
Drawing a point $(x,y)$ from the initial
density $\rho=\rho_0$, defined by setting $\omega=\omega_0$, and
switching the frequency to $\omega_1>\omega_0$ instantaneously
amounts in the work $\frac{1}{2}(\omega_1^2-\omega_0^2)(x^2+y^2)$.
The probability density of observing a specific work value $w$ is
given by the exponential density $\pf$ with
$\mu_0=\frac{\omega_1^2-\omega_0^2}{\omega_0^2}$.
Switching the frequency in the reverse direction,
$\omega_1\rightarrow\omega_0$, with the point $(x,y)$ drawn
from $\rho=\rho_1$ with $\omega=\omega_1$, the density of work
(with interchanged sign) is given by $\pr$ with
$\mu_1=\frac{\omega_1^2-\omega_0^2}{\omega_1^2}=\frac{\mu_0}{1+\mu_0}$.
The free-energy difference of the states characterized by $\rho_0$
and $\rho_1$ is the log-ratio of their normalizing constants,
$\df=-\ln\frac{Z_1}{Z_0}=\ln(1+\mu_0)$.
A plot of the work densities for $\mu_0=10$ is enclosed in
Fig.~\ref{fig:2}. 

Now, with regard to free-energy estimation,  is it better to use
one- or two-sided
estimators? In other words, we want to know whether the global minimum
of $M(\al)$ is on the boundaries $\{0,1\}$ of $\al$ or not. By the
convexity of $M$, the answer is determined by the signs of the derivatives
$M'(0)$ and $M'(1)$ at the boundaries. The asymptotic mean square
errors \gl{Mf} and \gl{Mr} of the one-sided estimators are calculated to be
\begin{align}
  M(1) = \varf{e^{-w+\Delta f}} = \frac{\mu_0^2}{1+2\mu_0}
\end{align}
for the forward direction and
\begin{align}
  M(0) = \varr{e^{w-\Delta f}} = \frac{\mu_0^2}{1-\mu_0^2}, \quad \mu_0<1,
\end{align}
for the reverse direction. For $\mu_0\geq1$ the variance of the
reverse estimator diverges.
Note that $M(0)>M(1)$ holds for all $\mu_0>0$, i.e.\ forward estimation of
$\df$ is always superior if compared to reverse estimation.
Furthermore, a straightforward calculation gives
\begin{align}
  M'(1) = \frac{\mu_0^3(\mu_0+\xi_-)(\mu_0-\xi_+)}{(1+2\mu_0)^2(1+3\mu_0)},
\end{align}
where $\xi_{\pm}=\frac{1}{2}(\sqrt{17}\pm3)$, and
\begin{align}
  M'(0) =
  -\frac{\mu_0^3\left(2+(1-2\mu_0)\mu_0\right)}{(1-\mu_0^2)^2(1-2\mu_0)},
  \quad \mu_0<\frac{1}{2},
\end{align}
and $M'(0)=-\infty $ for $\mu_0\geq\frac{1}{2}$. Thus, for the range
$\mu_0\in(0,\xi_+)$ we have $M'(0)<0$ as well as $M'(1)<0$ and therefore
$\al_o=1$, i.e.\ the forward estimator is superior to any two-sided
estimator \gl{benest} in this range.
For $\mu_0\in(\xi_+,\infty)$ we have $M'(0)<0$ and $M'(1)>0$,
specifying that $\al_o\in(0,1)$, i.e.\ two-sided estimation
with an appropriate choice of $\al$ is optimal.

\begin{figure}
  \includegraphics{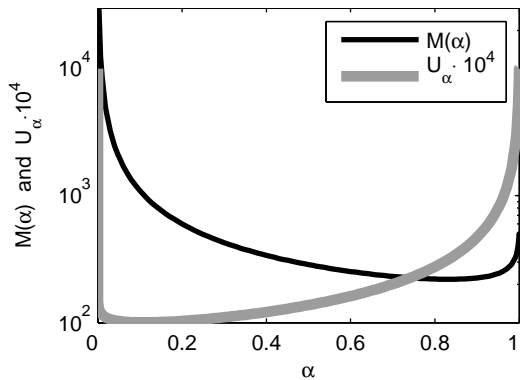}
  \caption{\label{fig:3} The overlap function $U_\al$ and the
    rescaled
    asymptotic mean square error $M$ for $\mu_0=1000$. Note that $M(\al)$
    diverges for $\al\rightarrow 0$.}
\end{figure}

Numerical calculation of the function $\Ua$ and subsequent evaluation
of $M(\al)$ allows to find the ``exact'' optimal fraction $\al_o$.
Examples for $\Ua$ and $M$ are plotted in Fig.~\ref{fig:3}.

\begin{figure}
  \includegraphics{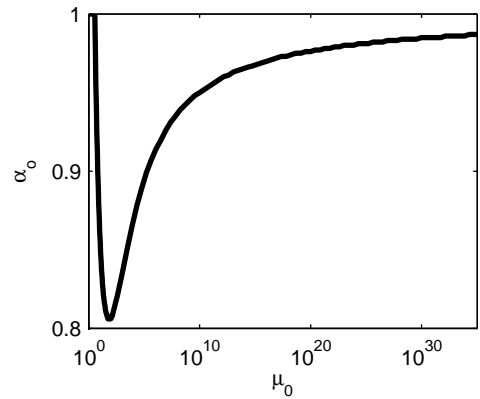}
  \caption{\label{fig:4} The optimal fraction
    $\al_o=\frac{n_0}{N}$ of forward work values for the two-sided
    estimation in dependence of the average forward work $\mu_0$.
    For $\mu_0\leq\xi_+\approx3.56$ the one-sided forward estimator is
    optimal, i.e.\ $\al_o=1$.}
\end{figure}

The behavior of $\al_o$ as a function of $\mu_0$ is quite
interesting, see Fig.~\ref{fig:4}. We can interpret
this behavior in terms of the Boltzmann
distributions as follows. Without loss
of generality, assume $\omega_0=1$ is fixed.
Increasing $\mu_0$ then means increasing
$\omega_1$. The density $\rho_1$ is fully nested in $\rho_0$,
cf.\ the inset of Fig.~\ref{fig:2} (remember that
$\omega_1>\omega_0$) and converges to a delta-peak
at the origin with increasing $\omega_1$.
This means that by sampling from $\rho_0$ we can obtain information
about the full density $\rho_1$ quite easily, whereas sampling
from $\rho_1$ provides only poor information about
$\rho_0$. This explains why $\al_o=1$
holds for small values of $\mu_0$. However, with
increasing $\omega_1$ the density $\rho_1$
becomes so narrow that it becomes difficult to obtain
draws from $\rho_0$ that fall into
the main part of $\rho_1$. Therefore, it is better
to add some information from
$\rho_1$, hence, $\al_o$ decreases. Increasing
$\omega_1$ further, the relative
number of draws needed from $\rho_1$ will decrease,
as the density converges towards
the delta distribution. Finally, it will become sufficient to make
only \textit{one} draw from $\rho_1$
in order to obtain the full information available.
Therefore, $\al_o$ converges towards
$1$ in the limit $\mu_0\rightarrow\infty$.

\begin{figure}
  \includegraphics{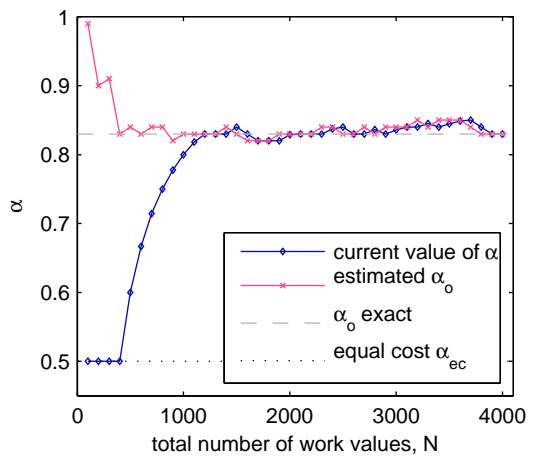}
  \caption{\label{fig:5} (Color online) Example of a single run using the
    dynamic strategy: the optimal fraction $\al_o$ of forward measurements
    for the two-sided free-energy estimation is estimated at
    predetermined values of total sample sizes $N=n_0+n_1$ of
    forward and reverse work values. Subsequently,
    taking into account the current actual fraction $\al=\frac{n_0}{N}$,
    additional work values are drawn such that we come
    closer to the estimated $\alhat_o$.}
\end{figure}

\begin{figure}
  \includegraphics{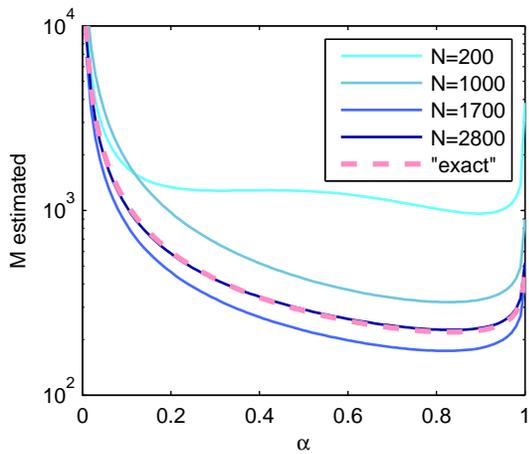}
  \caption{\label{fig:6} (Color online) Displayed are estimated mean square
    errors $\Mhat$
    in dependence of $\al$ for different sample sizes.
    The global minimum of the estimated function $\Mhat$ determines the
    estimate of the optimal fraction $\al_o$ of forward work
    measurements.}
\end{figure}

\begin{figure}
  \includegraphics{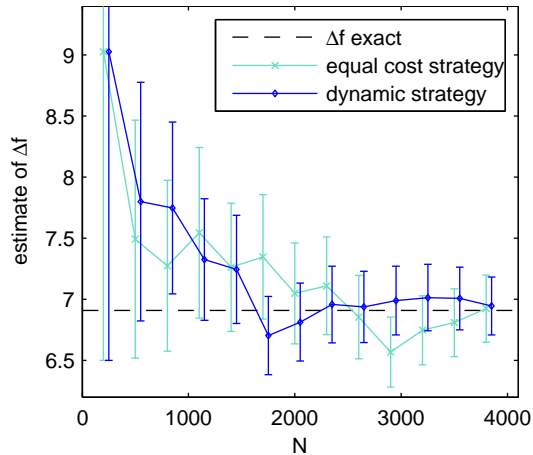}
  \caption{\label{fig:7} (Color online) Comparison of a single run of
    free-energy
    estimation using the equal cost strategy versus a single run
    using the dynamic strategy. The errorbars are the square roots
    of the estimated mean square error $X$.}
\end{figure}

In the following the dynamic strategy proposed in
Sec.~\ref{sec:7} is applied.
We choose $\mu_0=1000$ and $c_0=c_1$. The equal cost strategy draws
according to $\al_{ec}=0.5$ which is used as initial value in the
dynamic strategy. The results of a single run
are presented in Figs.~\ref{fig:5}--\ref{fig:7}. Starting with $N=100$,
the estimate of $\al_o$ is updated in steps of $\Delta N=100$. The actual
forward fractions $\al$ together with the estimated values of the optimal
fraction $\al_o$ are shown in Fig.~\ref{fig:5}.
The first three estimates of $\al_o$ are rejected, because the estimated
function $\Mhat(\al)$ is not yet convex.
Therefore, $\al$ remains unchanged at the beginning. Afterwards, $\al$
follows the estimates of $\al_o$ and starts to fluctuate about the
``exact'' value of $\al_o$.
Some estimates of the function $M$ corresponding to
this run are depicted in Fig.~\ref{fig:6}.
For these estimates $\al$ is discretized in
steps $\Delta\al=0.01$\,. Remarkably, the
estimates of $\al_o$ that result from these curves
are quite accurate even for relatively small $N$.
Finally, Fig.~\ref{fig:7} shows
the free-energy estimates of the run (not for all
values of $N$), compared with
those of a single run where the equal cost strategy is used.
We find some increase of accuracy when using the dynamic strategy.

\begin{figure}
  \includegraphics{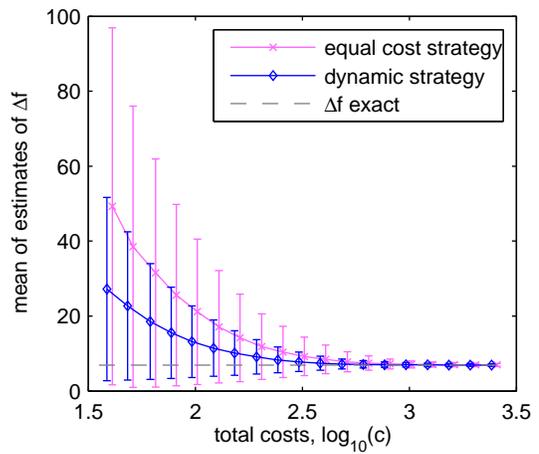}
  \caption{\label{fig:8} (Color online) Averaged estimates from $10\,000$
    independent runs with dynamic strategy versus $10\,000$ runs with
    equal cost strategy in dependence of
    the total cost $c=n_0c_0+n_1c_1$ spend. The cost ratio is $c_1/c_0=0.01$,
    $c_0+c_1=2$, and $\mu_0=1000$. The errorbars represent one standard
    deviation. Here, the initial value of $\al$ in the dynamic strategy is
    $0.5$, while the equal cost strategy draws with $\al_{ec}\approx0.01$.
    We note that $\al_o\approx0.08$.}
\end{figure}

\begin{figure}
  \includegraphics{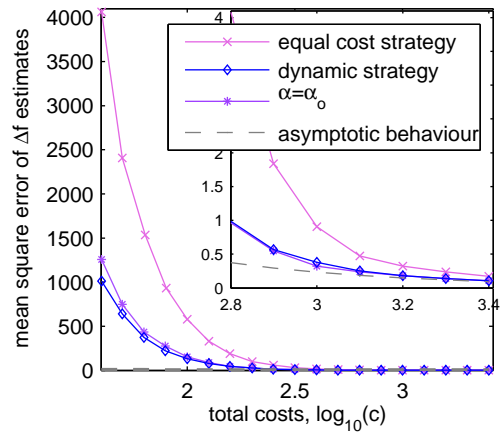}
  \caption{\label{fig:9} (Color online) Displayed are mean square errors
    of free-energy estimates using the same data as in
    Fig.~\ref{fig:8}. In addition, the mean square errors of
    estimates with constant $\al=\al_o$ are included, as well as the
    asymptotic behavior, Eq.~\gl{Xca}.
    The inset shows that the mean square error of the
    dynamic strategy approaches the asymptotic optimum, whereas the equal
    cost strategy is suboptimal. Note that for small sample sizes the
    asymptotic behavior does not represent the actual mean square
    error.}
\end{figure}

In combination with a good a priori choice of the
initial value of $\al$, the use of the dynamic strategy
enables a superior convergence and precision of free-energy
estimation, see Figs.~\ref{fig:8} and \ref{fig:9}.
Due to insight into some particular system under consideration, it is
not unusual that one has a priori knowledge which results in a better
guess for the initial choice of $\al$ in the dynamic
strategy than starting with $\al=\al_{ec}$.
For instance, a good initial choice
is known when estimating the chemical potential via Widom's
particle insertion and deletion \cite{Widom1963}. Namely, it is
a priori clear that inserting particles yields much more information
then deleting particles, since the phase-space which is accessible
to particles in the ``deletion-system" is effectively contained
in the phase-space accessible to the particles in the ``insertion-system",
cf.\ e.g.~\cite{Hahn2009}. A good a priori initial choice for $\al$ may be
$\al=0.9$ with which the dynamic strategy outperforms any other
strategy that the authors are aware of.

Once reaching the limit of large sample sizes, the dynamic strategy
is insensitive to the initial choice of $\al$, since the
strategy is robust and finds the optimal fraction $\al_o$
of forward measurements itself.

\section{Conclusion}\label{sec:9}

Two-sided free-energy estimation, i.e.\ the acceptance ratio method
\cite{Bennett1976},
employs samples of $n_0$ forward and $n_1$ reverse work measurements
in the determination of free-energy differences in a statistically
optimal manner.
However, its statistical properties depend strongly on the
ratio $\frac{n_1}{n_0}$ of work values used.
As a central result we have proven the convexity
of the asymptotic mean square error of two-sided free-energy
estimation as a function of the fraction $\al=\frac{n_0}{N}$
of forward work values
used. From here follows immediately the existence and uniqueness of the
optimal fraction $\al_o$ which minimizes the asymptotic mean square error.
This is of particular interest if we can
control the value of $\al$, i.e.\ can make additional measurements
of work in either direction. Drawing such that we finally reach
$\frac{n_0}{N}=\al_o$, the efficiency of two-sided estimation
can be enhanced considerably. Consequently, we have developed
a dynamic sampling strategy which iteratively estimates $\al_o$
and makes additional draws or measurements of work. Thereby,
the convexity of the mean square error enters as a key criterion
for the reliability of the estimates.
For a simple example which allows to compare with analytic
calculations, the dynamic strategy has shown to work perfectly.

In the asymptotic limit of large sample sizes the dynamic strategy is
optimal and outperforms any other strategy. Nevertheless, in this
limit it has to compete with the near optimal equal cost
strategy of Bennett which also performs very good. It is worth mentioning
that even if the latter comes close to the performance of ours, it is
worthwhile the effort of using the dynamic strategy, since
the underlying algorithm can be easily implemented and does cost quite
anything if compared to the effort required for drawing additional work
values.

Most important for experimental and numerical estimation of free-energy
differences is the range of small and moderate sample sizes.
For this relevant range, it is found that the dynamic strategy
performs very good, too. It converges significantly better
than the equal cost strategy. In particular, for small and moderate sample
sizes it can improve the accuracy of free-energy estimates by half an order
of magnitude.

We close our considerations by mentioning that the two-sided estimator
is typically far superior with respect to one-sided estimators:
assume the support and $\pf$ and $\pr$ is symmetric about $\df$ \footnote{
which is not the case for the densities studied in Sec.~\ref{sec:8}};
then, if the densities are symmetric to each other, $\pf(\df+w)=\pr(\df-w)$,
the optimal fraction of forward draws is $\frac{n_0}{N}=\frac{1}{2}$
by symmetry. Therefore, if the symmetry is violated not too strongly,
the optimum will remain near $0.5$\,. Continuous deformations of the
densities change the optimal fraction $\al_o$ continuously. Thus, $\al_o$
does not reach $0$ and $1$, respectively, for some certain strength of
asymmetry. It is exceptionally hard to violate the symmetry such that
$\al_o$ hits the boundary $0$ or $1$. In consequence, in almost
all situations, the two-sided estimator is superior.

\section{Acknowledgments}\label{sec:10}

We thank Andreas Engel for a critical reading of the manuscript.

\end{document}